\title{Local Correction with Constant Error Rate}
\author{
Noga Alon\thanks{Sackler School of Mathematics
and Blavatnik School of Computer Science,
Tel Aviv University,
Tel Aviv 69978, Israel.
Email: {\tt nogaa@tau.ac.il}.
Research supported in part by an ERC Advanced
grant, by a USA-Israeli BSF grant and by
the Israeli I-Core program.}
\and
Amit Weinstein\thanks{
Blavatnik School of Computer Science, Tel Aviv University, Tel Aviv
69978, Israel. Email: {\tt amitw@tau.ac.il}. Research supported in part by
an ERC Advanced grant and by the Israeli Centers of Research
Excellence (I-CORE) program.} }

\documentclass [11pt]{article}
\usepackage[]{amsfonts,amsmath,amssymb,amsthm,hyperref}
\usepackage[usenames]{color}
\usepackage{algorithm}
\usepackage[noend]{algorithmic}

\newtheorem{theo}{Theorem}[section]
\newtheorem{prop}[theo]{Proposition}
\newtheorem{lemma}[theo]{Lemma}

\newtheorem*{question*}{Question}
\newtheorem*{remark}{Remark}
\newtheorem{definition}[theo]{Definition}
\newtheorem*{definition*}{Definition}

\newcommand{\Exp}[1]{\mathbf{E}\left[#1\right]}
\newcommand{\Expx}[2]{\mathbf{E}_{#1}\left[#2\right]}

\newcommand{\ZZ}{\mathbb{Z}}

\newcommand{\Inf}{\mathrm{Inf}}

\newcommand{\SymInf}{\mathrm{SymInf}}

\newcommand{\core}{\mathrm{core}}

\newcommand{\calI}{\mathcal{I}}
\newcommand{\calJ}{\mathcal{J}}

\newcommand{\calS}{\mathcal{S}}

\newcommand{\calDI}{\mathcal{D_I}}
\newcommand{\calDIz}{\mathcal{D}_{\mathcal{I}_0}}
\newcommand{\calDIo}{\mathcal{D}_{\mathcal{I}_1}}
\newcommand{\calDIzo}{\mathcal{D}_{\mathcal{I}_0 \cup \mathcal{I}_1}}
\newcommand{\calDIW}{\mathcal{D}_{\mathcal{I}}^{W}}

\def \eps {\varepsilon}

\begin{document}
\maketitle

\begin{abstract}
A Boolean function $f$ of $n$ variables is said to be $q$-locally
correctable if, given a black-box access to a function $g$ which is
"close" to an isomorphism $f_{\sigma}(x)=f_{\sigma}(x_1, \ldots, x_n) = f(x_{\sigma(1)}, \ldots, x_{\sigma(n)})$
of $f$, we can compute
$f_{\sigma}(x)$ for \emph{any} $x \in \ZZ_2^n$ with good probability
using $q$ queries to $g$. It is known that degree $d$
polynomials are $O(2^d)$-locally correctable, and that most
$k$-juntas are $O(k \log k)$-locally correctable, where the closeness
parameter, or more precisely the distance between $g$ and $f_{\sigma}$,
is required to be exponentially small (in $d$ and $k$ respectively).

In this work we relax the requirement for the closeness parameter
by allowing the distance between the functions to be a constant.
We first investigate the family of juntas, and show that almost every $k$-junta
is $O(k \log^2 k)$-locally correctable for any distance $\eps < 0.001$.
A similar result is shown for the family of partially symmetric functions,
that is functions which are indifferent to any reordering of all but a constant
number of their variables. For both families, the algorithms provided here
use non-adaptive queries and are applicable to most but not all functions
of each family (as it is shown to be impossible to locally correct all of them).

Our approach utilizes the measure of symmetric influence
introduced in the recent analysis of testing partial symmetry of functions.
\end{abstract}

\section{Introduction}

Local correction of functions deals with the task of determining the value of a function in a
given point by reading its values in several other points. More precisely,
we care about locally correcting specific functions which are known up
to isomorphism, that is, functions which are known up to reordering of the input variables.
Our main interest is identifying the number of needed queries for this task, for a given function.
For a permutation $\sigma \in \calS_n$ and a function $f = f(x_1, \ldots, x_n): \ZZ_2^n \to \ZZ_2$,
let $f_{\sigma}$ denote the function given by $f_{\sigma}(x_1, \ldots, x_n) = f(x_{\sigma(1)}, \ldots, x_{\sigma(n)})$.

\begin{question*} Given a specific Boolean function $f$, what is the needed
query complexity in order to correct an input function which is
\emph{close} to some isomorphism $f_{\sigma}$ of $f$?
\end{question*}

This question can be seen as a special case of locally correctable
codes (see, e.g.,~\cite{Y}). Each codeword would be the $2^n$ evaluations of an
isomorphic copy of the input function, and thus the number of distinct codewords is at most $n!$,
and we would like to correct any specific value of the given noisy codeword using as
few queries as possible.

The notion of closeness in the above question plays a crucial role in answering it.
We say that two functions are $\eps$-close to one another if they differ
on at most an $\eps$ fraction of the inputs. Equivalently, $f$ is $\eps$-close to $f'$ if
$\Pr_x [ f(x) \neq f'(x) ] \leq \eps$, over a uniformly chosen $x \in \ZZ_2^n$.
The main focus of this work is to better understand the functions for which a constant
number of queries suffices for local correction, while we allow $\eps$ to be a constant as well.
In particular, we show that for partially symmetric functions, that is, functions which are symmetric
with respect to all but a constant number of their variables, this is typically the case.

The field of local correction of Boolean functions is closely related to that of
property testing, and in particular to testing isomorphism of functions.
In testing, the goal is to distinguish between a function which satisfies
some property and functions which are far from satisfying that property,
while here we are guaranteed the input function is close to satisfy a property, the property of being isomorphic
to some specific function, and we are required to locally correct a given input.
Due to this resemblece, many tools used in the research of local correction are borrowed
from the field of testing functions isomorphism and property testing in general
(see e.g.~\cite{CG, FKRSS, AKKLR, BO, AB, CGM}).

\subsection{Preliminaries}

Below is the formal definition of locally correctable functions, as given in~\cite{AW}.

\begin{definition*}
A Boolean function $f: \ZZ_2^n \to \ZZ_2$ is $q$-locally
correctable for $\eps > 0$ if the following holds. There exists an
algorithm that given an input function $g$ which is $\eps$-close to
an isomorphism $f_{\sigma}$ of $f$, can determine the value
$f_{\sigma}(x)$ for \emph{any} specific $x \in \ZZ_2^n$ with
probability at least $2/3$, using $q$ queries to $g$.
\end{definition*}

The two interesting parameters of the above definition are the noise rate $\eps$ and 
the number of queries $q$. In our recent work~\cite{AW} it is shown that when we
allow the noise rate to be relatively small, depending on the structure of the function,
functions from several interesting classes can be locally corrected. More precisely, it is observed that
every degree $d$ polynomial can be locally corrected from noise rate of $\eps < 2^{-d-3}$
using $O(2^d)$ queries. The family of $k$-juntas, that is, functions which only depend on $k$
of their input variables, are in particular degree $k$ polynomials and hence the same upper bound
is applicable. However, this is usually not tight as most $k$-juntas can actually be locally corrected
from the same noise rate using $O(k\log k)$ queries.
For more details on the above results, we refer the reader to~\cite{AW}.

The above results require an exponential dependency between the noise rate
$\eps$ and the parameter determining the structure of the function, such as the size of the junta or the degree of the polynomial.
This raises the following natural question.
Can we locally correct functions from these families when the noise rate is higher
or even constant? More generally, which functions can or cannot be locally corrected with a constant
number of queries from a given constant noise rate?

\subsection{Our results}

The main result of this work is the identification of two families, in which most functions can be locally
corrected from a constant noise rate. The first family is that of juntas.
We say that almost every $k$-junta satisfies a property if only $\eps_k$ fraction of the $k$-juntas
do not satisfy it, where $\eps_k$ tends to zero as $k$ tends to infinity.

\begin{theo}\label{thm:juntas}
Almost every $k$-junta can be locally corrected from a noise rate of
$\eps = 0.001$, using $O(k \log^2 k)$ non-adaptive queries.
\end{theo}

Similarly to the result in~\cite{AW}, this statement applies to almost every junta (and not to all of them).
The main differences between the results
are the constant noise rate which does not depend on the junta size $k$,
and the fact that the algorithm we describe is non-adaptive (at the expense of
increasing the query complexity by a logarithmic factor).

The second main result presented here is an extension of Theorem~\ref{thm:juntas} to another family of functions,
the family of partially symmetric functions, as defined in~\cite{BWY}. A function $f$ is called $t$-symmetric
if there exists a set of $t$ variables such that $f$ is symmetric with respect to these variables (that is,
any reordering of these variables does not change the function).
To better see the connection between these functions and juntas,
the following equivalent definition is often useful.
We say that $f$ is $(n-k)$-symmetric if there exists a set
of $k$ variables such that the output of $f$ is determined by these $k$ variables,
and the Hamming weight of the others. A $k$-junta is in particular an $(n-k)$-symmetric function,
and hence the following theorem can be viewed as a generalization of Theorem~\ref{thm:juntas}.

\begin{theo}\label{thm:psfs}
Almost every $(n-k)$-symmetric function can be locally corrected
from a noise rate of $\eps = 0.001$ using $O(k \log^2 k)$
non-adaptive queries.
\end{theo}

Here too the term "almost every" means that only $\eps_n$ fraction of these functions
do not satisfy the above where $\eps_n$ tends to zero as $n$ tends to infinity, and it
does not depend on $k$.

The proof of the theorem relies on the analysis of partially symmetric functions and borrows some
of the ideas provided in~\cite{BWY}. Notice that although juntas are in particular partially symmetric
functions, Theorem~\ref{thm:juntas} is not a corollary of Theorem~\ref{thm:psfs}, as juntas represent a
small fraction of all partially symmetric functions, and these results are applicable only to most functions in
the respective families.

The above theorems apply to almost every junta and partially symmetric function but not to all of them.
As the next simple result indicates, this restriction is unavoidable. Some functions in these families
are not locally correctable from a constant noise rate, regardless of the number of queries.

\begin{prop}\label{prop:impossible-juntas-psfs}
For every constant $\eps > 0$, there exists $k(\eps)$ such that the following holds.
For every $k \geq k(\eps)$, there exist $k$-juntas which cannot be locally corrected from $\eps$ noise rate,
regardless of the number of queries.
\end{prop}

\begin{proof}
Fix some $\eps > 0$ and let $k(\eps) = \lceil \log 1/\eps \rceil$.
Given some $k \geq k(\eps)$, we consider the $k$-junta $f$ which is defined to be 1 only when the first
variable is 1 and the other $k-1$ variables following it are 0. This function is obviously a $k$-junta as it
is determined by the first $k$ variables only. Notice however that the constant zero function is
$2^{-k} \leq \eps$ close to $f$, and therefore if we try to locally correct it, we will not be able to identify which
isomorphism of $f$ we were given. Hence we would not be able to locally correct $f$, regardless of the
number of queries.
\end{proof}

The above result is rather extreme, in the sense that we have no way of identifying
which original isomorphism was chosen. For most functions this is not the case, and this allows us
to achieve the previous results. The last result we present in this work shows that locally
correcting most functions is relatively hard, even from the smallest possible error rate of a single error.

\begin{theo}\label{thm:most-functions}
Almost every function over $n$ variables cannot be locally corrected, even from a single error (i.e., $\eps = 2^{-n}$),
using fewer than $n/100$ non-adaptive queries.
\end{theo}

The proof of Theorem~\ref{thm:most-functions} appears in Section~\ref{sec:conclusions} along with several
open questions. In Sections~\ref{sec:juntas} and~\ref{sec:psfs} we prove
Theorems~\ref{thm:juntas} and~\ref{thm:psfs}, respectively. The two proofs share a similar structure.

\section{Correcting juntas}
\label{sec:juntas}

Our approach for locally correcting juntas consists of two main
steps. The goal of the first step is to identify the junta
variables, i.e. those variables which determine the output of the
function. Since our query complexity should not depend on the input
size, one cannot hope to recover their exact location. Instead, we
use the testing-by-implicit-learning approach (see, e.g., \cite{Rocco}) and only identify
large sets which contain these variables. The second step, performed
after we have identified $k$ sets, each of which containing one of the junta
variables, is recovering their internal order (out of the $k!$
possible orderings). Once both these steps are completed, we would be able to
output the correct value of the function for the requested input.

In Section~\ref{sec:juntas-tools} we define some of the tools and
typical properties of juntas. The two steps of the algorithm are
later described in Secctions~\ref{sec:juntas-find-inf-sets} and
\ref{sec:juntas-alg}, completing the proof of
Theorem~\ref{thm:juntas}.

\subsection{Properties of juntas}\label{sec:juntas-tools}

Since juntas depend on a relatively small number of variables, we
often consider their concise representation over these variables
only. Given a $k$-junta $f$, we denote the \emph{core} of $f$ by
$f_{\core}: \ZZ_2^k \to \ZZ$, which is the function $f$ restricted
to its $k$ junta variables in their natural order.

A variable of a function is said to be \emph{influencing} if modifying its
value can modify the output of the function. Clearly a $k$-junta has
at most $k$ influencing variables, which are in fact the junta
variables (those which appear in its core). The following definition
quantifies how influential a variable, or more generally a set of
variables, is with respect to a given function.

\begin{definition}[Influence]
Given a Boolean function $f: \ZZ_2^n \to \ZZ_2$, the
\emph{influence} of a set of variables $J \subseteq [n] := \{1, 2, \ldots, n\}$ with
respect to $f$ is defined by
$$
\Inf_f(J) = \Pr_{x,y} \left[ f(x) \neq f(x_{\overline{J}}y_J)
\right]
$$
where $x_{\overline{J}}y_J$ is the vector whose $i$th coordinate
equals to $y_i$ if $i \in J$, and otherwise equals to $x_i$. When
the set $J = \{i\}$ is a singleton, we simply write $\Inf_f(i)$.
\end{definition}

An important property of influence is monotonicity. Namely, it is
known (see, e.g., \cite{FKRSS}) that for any two sets $J \subseteq
K$ and any function $f$, $\Inf_f(J) \leq \Inf_f(K)$. Given this
property, a set which has even a single variable with large
influence must also have large influence. We heavily rely on this
fact in our algorithm.

The result we present in this work is only applicable to most
juntas. The following two propositions indicate two typical
properties of juntas, which are required for our algorithm to work
with high probability. The first, presented in
Proposition~\ref{prop:juntas-high-inf}, indicates that in a typical
junta every influencing variable has constant influence. The second
property bounds the distance between a typical junta and its
isomorphisms, and is presented in Proposition~\ref{prop:juntas-no-iso}.
Notice that in both propositions, it suffices to consider the core
of the junta rather than the entire function.

\begin{prop}\label{prop:juntas-high-inf}
Let $f: \ZZ_2^k \to \ZZ_2$ be a random core of a $k$-junta. Then
with probability at least $1-2^{-\Omega(k)}$, any variable $i \in
[k]$ out of the $k$ variables of $f$ has influence $\Inf_f(i) >
0.1$.
\end{prop}

\begin{proof}
Let $f$ be a random function over $k$ variables. The influence of
some variable $i$ is determined by the number of pairs of inputs,
which differ only on the coordinate $i$, that disagree on the
output. Since each output of the function is chosen independently,
and as these $2^{k-1}$ pairs are disjoint, this is in fact a
binomial random variable. The influence of variable $i$ is
less than $0.1$ only if at most $1/5$ of these pairs disagree.
This probability is thus
$$
\Pr [B(2^{k-1}, 0.5) <  \tfrac{1}{5} \cdot 2^{k-1} ] < 2^{-c2^k}
$$
for some absolute constant $c > 0$, where here B is the binomial
distribution and we applied one of the standard estimates for
binomial distributions (cf., e.g. \cite{AS}, Appendix A). Therefore,
by the union bound, all k variables have influence greater
than 0.1 with probability $1-2^{-\Omega(k)}$.
\end{proof}

\begin{prop}\label{prop:juntas-no-iso}
Let $f: \ZZ_2^k \to \ZZ_2$ be a random core of a $k$-junta. Then
with probability at least $1-2^{-\Omega(k)}$, $f$ is $0.1$-far from
any non-trivial isomorphic function.
\end{prop}

\begin{proof}
Let $f$ be a random function of $k$ variables and let $\pi \in
\calS_k$ be any non-trivial permutation. Our goal is to show that
$\Pr_x [f(x) \neq f_{\pi}(x)] > 0.1$ for every such $\pi$, with high
probability (where the probability is over the choice of $f$ and
applies to all permutations simultaneously). We will do a similar
calculation to the one above. Here however, we do not have a nice
partition of the inputs into disjoint pairs, as some inputs remain
unchanged after applying the permutation. Consider the partition of
the inputs according to $\pi$ into chains, that is elements $x, \pi
x, \pi^2 x, \ldots, \pi^i x = x$. Notice that if $\pi$ is not the
identity, there are at most $2^k/2$ chains of length 1 (half of
the elements). Looking at the elements of a chain of length $i \geq
2$, $i-1 \geq \lceil i/2 \rceil $ of them result in pairs $x, \pi x$ so that
all these events $f(x) \neq f(\pi x)$ are mutually independent.
Thus in total we have at
least $2^k/4$ independent samples.

As before, we can now bound the probability that
$\Pr_x [f(x) \neq f_{\pi}(x)] < 0.1$ by the probability that at most
$2/5$ of these pairs would disagree (as with probability at least
$1/4$ we fall into an element from our independent samples). We
bound this probability by
$$
\Pr [B(2^{k}/4, 0.5) <  \tfrac{2}{5} \cdot \tfrac{2^{k}}{4} ] <
2^{-c'2^k}
$$
for some absolute constant $c' > 0$, where again we applied one of
the standard estimates for binomial distributions. Since there are
only $k!-1$ non-trivial permutations, we can apply the union bound and
conclude that over the choice of $f$, it is $0.1$-far from all its
non-trivial isomorphic copies with probability at least
$1-2^{-\Omega(k)}$.
\end{proof}

\subsection{Finding the influencing sets}\label{sec:juntas-find-inf-sets}

In order to find the influencing sets, we first need a way to
estimate the influence of a set by querying the input function. To
this end we use the following natural algorithm
\textsc{Estimate-Influence}.

\begin{algorithm}[tbh]
  \caption{\textsc{Estimate-Influence}$(f, J, \delta, \eta)$}
  \begin{algorithmic}[1] \label{alg:est-inf}
   \STATE Set $q = \lceil \tfrac {\ln 2/\eta}{2\delta^2} \rceil$ and $X = 0$.
   \FOR{$i=1$ to $q$}
    \STATE Pick two random inputs $x,y \in \ZZ_2^n$.
    \STATE Increase $X$ by 1 if $f(x) \neq f(x_{\overline{J}}y_J)$.
   \ENDFOR
   \STATE Return $X/q$.
  \end{algorithmic}
\end{algorithm}

\begin{prop}
For any function $f$, a set of variables $J$ and two constants
$\delta, \eta \in (0,1)$, the algorithm
\textsc{Estimate-Influence}$(f, J, \delta, \eta)$ returns a
value within distance $\delta$ of $\Inf_f(J)$ with probability at
least $1-\eta$, by performing $O(\delta^{-2} \log 1/\eta)$
non-adaptive queries to $f$.
\end{prop}

\begin{proof}
The proof is a direct application of the Chernoff bound as
$\Exp{X/q} = \Inf_f(J)$ and we deviate by more than $\delta$
with probability at most $2\exp(-2\delta^2q)$.
\end{proof}

In our scenario, however, we also need to consider the fact that we
query a noisy version of the function. The following proposition
shows that the noise cannot modify the influence of a set by too
much, and therefore we can still estimate correctly which sets have
significant influence.

\begin{prop}\label{prop:inf-distance}
Let $f$ and $g$ be any two functions which are $\eps$-close. Then
for every set $J \subseteq [n]$ of variables, $| \Inf_f(J) -
\Inf_g(J) | \leq 2\eps$.
\end{prop}

\begin{proof}
Fix $J$ to be some subset of the variables and let $f$ and $g$ be
two functions which are $\eps$-close. We prove that $\Inf_f(J) \leq
\Inf_g(J) + 2\eps$, from which the proposition follows by simply
replacing the roles of $f$ and $g$. By the triangle inequality,
\begin{eqnarray*}
 \Inf_f(J)
  & = & \Pr_{x,y} \left[ f(x) \neq f(x_{\overline{J}}y_J) \right] \\
  & \leq &
    \Pr_{x} \left[ f(x) \neq g(x) \right] +
    \Pr_{x,y} \left[ g(x) \neq g(x_{\overline{J}}y_J) \right] +
    \Pr_{x,y} \left[ g(x_{\overline{J}}y_J) \neq f(x_{\overline{J}}y_J) \right] \\
  & \leq & \Inf_g(J) + 2\eps \ .
\end{eqnarray*}
\end{proof}

We are now ready to describe the first step in our algorithm for
local correction of juntas, which is identifying the influencing
sets. A crucial restriction we must assume over the function we try
to correct is that the influence of any influencing variable is
significant enough, so we can identify them in spite of the noise
(namely, it should satisfy Proposition~\ref{prop:juntas-high-inf}).

The algorithm \textsc{Find-Influencing-Sets} is given a function
$f$, a partition of the variables $\calI$, a size parameter $k$ and
a noise parameter $\eps$. The algorithm returns all parts in the
partition $\calI$ which it considered as influential with respect to
$\eps$. We will later see that in the scenario we apply it, the
returned sets are exactly those which contain significantly
influencing variables of $f$, with high probability.

\begin{algorithm}[tbh]
  \caption{\textsc{Find-Influencing-Sets}$(f, \calI = \{I_1, \ldots, I_s\}, k, \eps)$}
  \begin{algorithmic}[1] \label{alg:find-inf-sets}
   \STATE Fix $r = \lceil 12k \ln s \rceil $ and $S = [s]$.
   \FOR{each of $r$ rounds}
    \STATE Pick a random subset $T \subseteq [s]$ by including each index independently with probability $1/k$.
    \STATE Define $J = \cup_{i \in T} I_i$ to be the union of sets in $\calI$ according to
    the indices of $T$.
    \STATE If \textsc{Estimate-Influence}$(f, J, \epsilon, 1/20r) \leq 3\eps$, set $S = S \setminus T$.
   \ENDFOR
   \STATE Return $\{I_i\}_{i \in S}$
  \end{algorithmic}
\end{algorithm}

\begin{remark}
The algorithm \textsc{Find-Influencing-Sets} is very similar to the
algorithm \textsc{BlockTest} defined in
\cite{Blais-non-adaptive-juntas}. The main difference is in the
noise tolerance behavior. In the original algorithm a part $I_j$ was
marked as non influential, and was removed from $S$, only if its
estimated influence was precisely 0. Here however we mark such a
part as non-influential even if it has some influence, but as long
as our estimate for it is small enough.
\end{remark}

\begin{lemma}\label{lem:find-inf-sets}
Let $f$ be a $k$-junta whose influencing variables each has
influence of at least $6\eps$ for some $\eps>0$, and they are
separated by a partition $\calI$, of size $|\calI| = O(k^2)$,
$|\calI| > 5$. Then for every function $g$ which is $\eps$-close to
$f$, \textsc{Find-Influencing-Sets}$(g, \calI, k, \eps)$ returns
exactly the $k$ sets containing the influencing variables of $f$
with probability at least $9/10$, by performing $O(k \log^2
k/\eps^2)$ non-adaptive queries to $g$.
\end{lemma}

\begin{proof}
Fix $\eps > 0$ and let $f, k, \calI$ and $g$ be as described in the
lemma. We first note that the query complexity is indeed $O(k \log^2
k / \eps^2)$ as we have $r = O(k \log s) = O(k \log k)$ rounds,
assuming $s = O(k^2)$, and in each round we perform $O(\log s /
\eps^2) = O(\log k / \eps^2)$ queries. Moreover, the queries are all
non-adaptive as we only apply the \textsc{Estimate-Influence}
algorithm which is non-adaptive as well.

By the analysis of \textsc{Estimate-Influence} and the parameters we
provide it, we know it would deviate by more than $\eps$ with
probability at most $1/20r$. Since we invoke it once per round and
there are only $r$ rounds, by the union bound they would all deviate
by at most $\eps$ simultaneously with probability at least $19/20$.
Assuming this is indeed the case, what remains to be shown is that
every set containing an influencing variable would be returned, and
only those.

Let $I$ be a set containing an influencing variable of $f$. Since we
know each influential variable of $f$ has influence at least
$6\eps$, by monotonicity of influence we have $\Inf_f(J) \geq
\Inf_f(I) \geq 6\eps$ for any set $J$ such that $I \subseteq J$.
Moreover, since $g$ is $\eps$-close to $f$, by
Proposition~\ref{prop:inf-distance} we have $\Inf_g(J) \geq 4\eps$
for any such set $J$. As we assumed all calls to
\textsc{Estimate-Influence} deviated by at most $\eps$, we would not
flag the set $I$ as non-influential at any round (and thus it would
be returned).

Consider now the case that $I$ contains no influencing variable of
$f$. It suffices to show that at some round, the set $J$ would
contain $I$ but no other set $I'$ which contains an influencing
variable. If there was such a round, then $\Inf_g(J) \leq \Inf_f(J)
+ 2\eps = 2\eps$ and we would estimate it correctly to be at most
$3\eps$ by our assumption. Since there are at most $k$ sets with
influencing variables, at each round, $J$ would include $I$ and no
other set with influencing variables with probability at least
$(1/k)(1-1/k)^k \geq 1/4k$ for $k \geq 2$. We can now bound the
probability that the set $I$ would be incorrectly returned by
$(1-1/4k)^r\leq e^{-r/4k} \leq e^{-3\ln s} < \tfrac{1}{20s}$. By
applying the union bound over all the parts in $\calI$, we get that
with probability at least $19/20$, all sets without influencing
variables would not be returned.

Combining both our assumptions, each occurring with probability at
least $19/20$, we indeed showed that with probability at least
$9/10$, precisely the sets which contain the influencing variables
of $f$ would be returned.
\end{proof}

\subsection{The algorithm}\label{sec:juntas-alg}

Before we proceed to describe the complete algorithm, we need some
additional definitions. Since we have no intention to identify the
exact location of the influencing variables, and we only identify
sets which contain them, it is natural to query the function at
inputs which are constant over the sets in a given partition.
For this purpose, we use the following distribution.

\begin{definition}
Let $\calI = \{I_1, \ldots, I_{2s}\}$ be a partition of $[n]$ into
an even number of parts (where some parts may be empty). The
distribution $\calDI$ over $y \in \ZZ_2^n$ is defined as follows.
\begin{itemize}
 \item Choose $z \in \ZZ_2^{2s}$ to be a random balanced vector of Hamming weight $|z| = s$.
 \item Define $y \in \ZZ_2^n$ such that for every $I_j \in \calI$ and $i \in I_j$, $y_i = z_j$.
\end{itemize}
\end{definition}

\begin{prop}[\cite{CGM}]\label{prop:junta-sample}
Let $J=\{j_1, \ldots, j_k\} \subseteq [n]$ be a set of size $k$, and let
$s = \Omega(k^2)$ be even. The distribution $\calDI$ satisfies the
following conditions.
\begin{itemize}
 \item For every $x\in\ZZ_2^n$, $\Pr_{\calI, y \sim \calDI} [ y = x ] =
 2^{-n}$ given that the partition $\calI$ was chosen at random.

 \item The marginal distribution of $y$ over the set of indices $J$ is
  $4k^2/s$-close to uniform over $\ZZ_2^{k}$ (in total variation distance),
  for a fixed partition $\calI$ which separates the variables of $J$.
\end{itemize}
\end{prop}

\begin{remark}
In our scenario, we sometimes use the distribution $\calDIzo$
where $\calI_0$ and $\calI_1$ are random partitions of $X_0$ and
$X_1$, such that $|\calI_0| = |\calI_1| = 2s$, and where $X_0 \cup
X_1 = [n]$ are a partition of $[n]$ into two parts. We define $y
\sim \calDIzo$ to be a merge of $y_0 \sim \calDIz$ and $y_1 \sim
\calDIo$ in the following way. The vector $y$ is the unique vector
for which $y_{X_0} = y_0$ and $y_{X_1} = y_1$. Thus, the first item
of Proposition~\ref{prop:junta-sample} holds as is, and in the
second item we have an additional factor of 2.
\end{remark}

The full algorithm is described as Algorithm~\ref{alg:local-corr-juntas} below.
From this point on, whenever we draw a random partition of some
set, we do so by assigning each element into one of the parts
uniformly and independently at random (which may result in some empty parts).
Recall that our goal is to locally correct a given function $g$,
which is $\eps$-close to $f_{\sigma}$, by returning the value
$f_{\sigma}(x)$ for the given input $x$. Additionally, $f$ is a
$k$-junta whose core is known to the algorithm, and we can and will
restrict ourselves to such functions which satisfy typical
conditions (namely Propositions~\ref{prop:juntas-high-inf} and
\ref{prop:juntas-no-iso}). Notice that in our scenario, the smaller
$\eps$ is the easier it is to correct so it suffices to show that
for $\eps=0.001$, the algorithm succeeds with good probability and
with the requested query complexity.

\begin{algorithm}[tbh] 
  \caption{\textsc{Locally-Correct-Junta}$(f_{\core}, k, g, x)$}
  \begin{algorithmic}[1] \label{alg:local-corr-juntas}
   \STATE Fix $s = 400k^2$ and $r = 2500 \lceil k\log k \rceil $.
   \FOR{$j \in \{0,1\}$}
     \STATE Let $X_j = \{ i \in [n] \mid x_i = j \}$.
     \STATE Randomly partition $X_j$ into $\calI_j$, consisting of (potentially) $s$ parts.
   \ENDFOR
   \STATE Define $\calI = \calI_0 \cup \calI_1$ to be the partition of $X_0 \cup X_1 = [n]$.
   \STATE Invoke \textsc{Find-Influencing-Sets}$(g, \calI, k, 0.01)$ and assign the result into $\calJ = \{ I_{a_1}, I_{a_2}, \ldots \}.$
   \STATE If $ | \calJ | \neq k$, or if $\calJ$ contains an empty set, return 0.
   \STATE Let $B = (b_1, \ldots, b_k)$ be an arbitrary ordered set such that $b_i \in I_{a_i}$ for every $i \in [k]$.
   \STATE For every $\ell \in [r]$, randomly sample $y^\ell \sim \calDIzo$ and query $g(y^\ell)$.
   \STATE Let $\pi \in \calS_k$ be the permutation which maximizes the number of indices
   $\ell$ for which\\ $g(y^\ell) = f_{\core}(y^\ell_{B_{\pi}})$, where $B_{\pi} = (b_{\pi(1)}, \ldots, b_{\pi(k)})$.
   \STATE Return $f_{\core}(x_{B_{\pi}})$
  \end{algorithmic}
\end{algorithm}

\begin{proof}[Proof of Theorem~\ref{thm:juntas}]
First, we analyze the query complexity of the algorithm. All the
queries the algorithm perform are by invoking
\textsc{Find-Influencing-Sets} and querying $y$ which was chosen
according to $\calDI$. In both cases, the queries only depend on the
partitions $\calI_0$ and $\calI_1$, and therefore they are
non-adaptive. The number of queries in these parts are $O(k\log^2k)$
and $O(k\log k)$ respectively, and thus the algorithm performs a
total of $O(k\log^2k)$ non-adaptive queries as required.

Let $f$ be a function which satisfies the conditions of
Propositions~\ref{prop:juntas-high-inf} and
\ref{prop:juntas-no-iso}. As each condition is satisfied with
probability at least $1-2^{-\Omega(k)}$, it suffices to show the
algorithm succeeds with high probability for such functions.

The success of the algorithm depends on the following three events.
The first, we need the partition $\calI$ to separate all the junta's
influencing variables. In each set of variables, $X_0$ and $X_1$,
there are at most $k$ influencing variables. Since we partition each
of them into $s$ parts at random, we would have such a bad collision
with probability at most $2{k \choose 2}/s < 1/30$ for our choice of
$s$.

The second event is identifying the correct sets, those for which $f_{\sigma}$ has influencing variables.
Assuming there were no collisions in the partition, by Lemma~\ref{lem:find-inf-sets} we will
identify precisely the influencing sets with probability at least $9/10$.

The third and last event which remains is correctly choosing the
permutation $\pi$. Here we rely on the fact that the core of $f$ is
not close to any non trivial permutation of itself. Since our
estimates are performed using queries according to $\calDIzo$, we
need to estimate how accurate they are. By
proposition~\ref{prop:junta-sample}, as each sample $y$ is
distributed uniformly, the distance between $g(y)$ and
$f_{\sigma(y)}$ is at most $0.001$ over the choice of $\calI_0$ and
$\calI_1$. By applying Markov's inequality, our partitions will
satisfy $\Pr_{y\sim \calDIzo}[g(y) \neq f_{\sigma}(y)] \leq 0.01$
with probability at least $9/10$.

Assume from this point on that the partitions $\calI_0$ and
$\calI_1$ were good, namely, satisfying the above inequality. By the
second part of proposition~\ref{prop:junta-sample}, the marginal
distribution of $y$ over the influencing variables, and hence over
the influencing sets (for each partition $\calI_j$), is
$4k^2/s$-close to uniform in total variation distance. Therefore,
when $\pi \in \calS_k$ is not the correct permutation we have
$$
\Expx{y \sim \calDIzo} { g(y) = f_{\core}(y_{B_{\pi}})  }
  \leq 1 - 0.1 + 0.01 + 2\tfrac{4k^2}{s} = 0.93 \ ,
$$
where $B_{\pi}$ are defined as in the algorithm. For the
correct permutation however, this expectation would be at least
$0.97$. Since our queries are independent, we can
apply the Chernoff bound. When we perform $r$ such queries, we
deviate from the expectation by at least 0.02 with probability at
most $\exp(-2\cdot 0.02^2 \cdot r) = \exp(-2\lceil k \log k
\rceil)$. By the union bound, our estimation for all $k!$
permutations is within 0.02 from the expectation with
probability at least $9/10$. Combining this with the probability
that the partitions are good, it follows that we choose the correct
permutation $\pi$ with probability at least $4/5$.

The failure probability of our algorithm can be bounded by the
probability that one of the above events does not occur. Therefore,
the algorithm fails with probability at most $1/30 + 1/10 + 4/5
= 1/3$, meaning it returns the correct answer with probability
at least $2/3$, as required.
\end{proof}

\begin{remark}
The assertion of the theorem holds also when the core of $f$ is isomorphic to
itself for some non-trivial permutation. The crucial requirement is
that it is not $\eps$-close to any of its permutations for
$0<\eps<0.1$. For simplicity, the proof does not consider these
cases, however this only influences the identification of the permutation $\pi$ which
is indifferent to which of the isomorphisms it corresponds.
\end{remark}

\section{Correcting partially symmetric functions}
\label{sec:psfs}

The algorithm for local correction of partially symmetric functions is similar
to the one just presented for locally correcting juntas.
The main tool for generalizing the algorithm and its proof is the
analogous measure to influence called \emph{symmetric influence},
introduced in \cite{BWY}. We repeat the definition of
symmetric influence and some of its properties in
Section~\ref{sec:psfs-tools}, where we also prove several properties of typical
partially symmetric functions.

In Section~\ref{sec:psfs-find-asym-sets} we describe the algorithms
for estimating the symmetric influence of a set and identifying all
the asymmetric sets in a partition (namely, sets which have large
symmetric influence). This is the first step in the correcting
algorithm. The second step, which is again recovering the ordering
of the identified sets, is described as part of the algorithm in
Section~\ref{sec:psfs-alg}, followed by the proof of
Theorem~\ref{thm:psfs}.

\subsection{Properties of partially symmetric functions}\label{sec:psfs-tools}

Unlike juntas, partially symmetric functions typically depend on all
the variables of the input. However, there is a large set of
variables which influence the function only according to their
combined Hamming weight, and not according to the value of each coordinate.
The concise representation of partially symmetric functions is
therefore different from that of juntas.

Let $f: \ZZ_2^n \to \ZZ_2$ be an $(n-k)$-symmetric function.
We define the core
of $f$ by $f_{\core}: \ZZ_2^k \times \{0,1, \ldots, n-k\} \to
\ZZ_2$, the function $f$ restricted to its $k$ asymmetric variables
and the Hamming weight of the remaining variables.

The notion of influence represents how much a variable, or a set of
variables, can influence the output of the function when their value
is modified. For partially symmetric functions however, as typically
all variables influence the output of the function, this is less
useful. Instead, as the \emph{special} variables are in fact the
asymmetric variables of the function, we find the following
definition of \emph{symmetric influence} more useful. We measure for
a set of variables, what is the probability that reordering them
would result in a change of the output of the function.

\begin{definition}[Symmetric influence, \cite{BWY}]
Given a Boolean function $f: \ZZ_2^n \to \ZZ_2$, the \emph{symmetric
influence} of a set of variables $J \subseteq [n]$ with respect to
$f$ is defined by
$$
\SymInf_f(J) = \Pr_{x \in \ZZ_2^n,\pi \in \calS_J} \left[ f(x) \neq
f(\pi x) \right] \ ,
$$
where $\calS_J$ is the set of permutations within $\calS_n$ which only
move elements inside the set $J$.
\end{definition}

Similar to juntas, a function $f$ is $(n-k)$-symmetric if and only if there
exists a set $K$ of size at most $k$ such that $\SymInf_f([n]
\setminus K) = 0$. The authors of \cite{BWY} defined
symmetric influence and showed it satisfies several properties which
are similar to those of influence. One of these properties is
monotonicity. For any two sets $J \subseteq K$ and any function $f$,
$\SymInf_f(J) \leq \SymInf_f(K)$. As in the case of juntas, our
algorithm heavily relies on this fact in order to identify the
asymmetric sets.

Like Theorem~\ref{thm:juntas}, Theorem~\ref{thm:psfs} is also applicable only to most partially
symmetric functions and not to all of them. We again define two
properties which are required for our algorithm to succeed with high probability, which are
typical to such functions. The first bounds the symmetric influence
of sets containing at least one asymmetric and one symmetric variable,
presented in Proposition~\ref{prop:psfs-high-syminf}.
The second, described in Proposition~\ref{prop:psfs-no-iso}, deals with
the distance between such a function and its non-trivial isomorphisms.

\begin{prop}\label{prop:psfs-high-syminf}
Let $f: \ZZ_2^n \to \ZZ_2$ be a random $(n-k)$-symmetric function for some $k < n$.
Then with probability at least $1-2^{-\Omega(\sqrt{n})}$, any asymmetric variable $i$
and any symmetric variables $j$ have symmetric influence
$\SymInf_f(\{i,j\}) > 0.1$.
\end{prop}

\begin{proof}
Let $f$ be a random $(n-k)$-symmetric function and assume without loss of generality that
its asymmetric variables are the first $k$ variables. We arbitrarily choose
the asymmetric variable $x_k$ and the symmetric variable $x_{k+1}$.
Let $x$ be a random vector and apply a random permutation on $x_k$ and $x_{k+1}$.
Notice that only with probability $1/4$ we might see two different outputs of $f$, as with
probability $3/4$ either $x_k = x_{k+1}$, or the
chosen permutation is the identity. We therefore restrict ourselves
to such inputs where $x_k \neq x_{k+1}$ and the permutation 
transposes $x_k$ and $x_{k+1}$.

We consider the partition of inputs to $f_{\core}$ into $2^{k-1} (n-k-1)$
pairs according to the value of $x_1,\cdots , x_{k-1}$ and the Hamming weight of the variables
$x_{k+2}, \ldots, x_n$. Notice that since we require $x_k$ to be different from $x_{k+1}$,
for each such restriction we have precisely two inputs to the core.
Moreover, the transposition of $x_k$ and $x_{k+1}$ only swaps
between the two inputs of the same pair, meaning the pairs are independent.
Define $m = n-k-2$ and for every $z \in \ZZ_2^{k-1}$ and $w \in \{0,1,\ldots, m\}$
let $X_{z,w}$ be the indicator random variable of the event that $f_{\core}$ agrees
on the corresponding pair of inputs. Using these definitions, we can compute
the symmetric influence of $k$ and $k+1$ as follows.
$$
\SymInf_f(\{k,k+1\}) = \frac{1}{8} + \frac{1}{8}
  \sum_{z \in \ZZ_2^{k-1}} \sum_{w=0}^{m}
  \frac{ {m \choose w} }{ 2^{k-1} 2^{m}} \cdot (-1)^{X_{z,w}} \ .
$$

To bound the deviation of the symmetric influence, we additionally define for each $z$ and $w$
the random variable $Y_{z,w} = {m \choose w} 2^{-(k-1)} 2^{-m} \cdot (-1)^{X_{z,w}}$.
The accumulated sum over $Y_{z,w}$ for every $z$ and $w$ is a martingale.
When we add a specific $Y_{z,w}$ to the sum, we modify it by ${m \choose w} 2^{-(k-1)} 2^{-m}$
in absolute value. Therefore, by the Azuma-Hoeffding inequality,
\begin{eqnarray*}
 \Pr[ \SymInf_f(\{k, k+1\}) < 0.1 ]
  & \leq &  \Pr[ | 8\cdot \SymInf_f(\{k, k+1\}) - 1 | > \tfrac{1}{5} ] \\
  & \leq & 2\exp\left(-\frac{ \tfrac{1}{5^2}  }{ 2 \sum_{z,w} |Y_{z,w}|^2  }\right) \\
  & = & 2\exp\left(-\frac{ 2^{2(k-1)}2^{2m}  }{ 2 \cdot 25 \cdot 2^{k-1} \sum_{w} {m \choose w}^2  }\right) \\
  & = & 2\exp\left(-\frac{ 2^{k}  }{ 100 } \cdot \frac{ 2^{2m} } {{2m \choose m}  }\right) \\
  & \approx & 2\exp\left(-\frac{ 2^{k}  }{ 100 } \cdot \sqrt{\pi m} \right) = 2^{-\Omega(2^k \sqrt{n-k})}\ ,
\end{eqnarray*}
where we used the known fact
$\sum_{i=0}^{m} {m \choose i}^2 = {2m \choose m} \approx \tfrac{2^{2m}}{\sqrt{\pi m}}$.

In order to complete the proof, we apply the union bound over the $k$ possible choices for the
asymmetric variable. Notice that for the symmetric variables it suffices to consider a single choice,
as they are symmetric. Thus the probability that any such symmetric influence would be smaller than
$0.1$ is bounded by $k \cdot 2^{-\Omega(2^k \sqrt{n-k})} = 2^{-\Omega(\sqrt{n})}$ as required.
\end{proof}

\begin{prop}\label{prop:psfs-no-iso}
Let $f: \ZZ_2^n \to \ZZ_2$ be a random $(n-k)$-symmetric function for some $k<n$,
and let $J$ be the set of its asymmetric variables.
Then with probability at least $1-2^{-\Omega(\sqrt{n})}$, $f$ is $0.1$-far from any non-trivial
permutation of itself which only moves elements within $J$.
\end{prop}

\begin{proof}
Let $f$ be a random $(n-k)$-partially symmetric function and assume
without loss of generality that its asymmetric variables are the first $k$ variables.
Our function $f$ is in fact a union of $n-k+1$ randomly chosen $k$-juntas
over the first $k$ variables, where the Hamming weight of the remaining varialbes
determines which junta we are invoking.
When $k \geq \sqrt{n}$, we can apply Proposition~\ref{prop:juntas-no-iso} over each of
these $k$-juntas. Each such function is $0.1$-far from being isomorphic to itself with probability
at least $1-2^{-\Omega(k)}$, and we apply the union bound over them.
Therefore, $f$ would be $0.1$-far from any non-trivial isomorphism of itself,
which only moves elements within the first $k$ variables,
with probability at least $1-(n-k+1)\cdot 2^{-\Omega(k)}=1-2^{-\Omega(\sqrt{n})}$.

Assume now that $k < \sqrt{n}$, and moreover that $k \geq 2$
(as when $k=0$ and $k=1$ the proposition trivially holds).
Let $f_0, \ldots, f_{n-k+1}$ denote the $k$-juntas representing $f$.
We fix some permutation $\pi \in \calS_k$ over the first $k$ variables.
As in the proof of Proposition~\ref{prop:juntas-no-iso},
for every fixed $w$ there are at least $2^k/4$ pairs of inputs $x, \pi x \in \ZZ_2^k$ for which
the values of the events $f_w(x) \neq f_w(\pi x)$ are independent.
By applying the Azuma-Hoeffding inequality as before for such an input $x$,
we have
$$
\Pr[ \Expx{w \sim B(n-k, 1/2)}{f_w(x) \neq f_w(\pi x)} < 0.45 ] < 2^{-\Omega(\sqrt{n-k})} = 2^{-\Omega(\sqrt{n})}\ .
$$
The probability of these events is very small, and therefore over our independent $2^k/4$ samples
we would have no more than $2^k/40$ of them occur, with probability at least $1-2^{-\Omega(\sqrt{n} \cdot 2^k)}$.
So overall for a given permutation $\pi$, the distance between $f$ and $f_{\pi}$ is at least
$\tfrac{1}{4} \cdot \tfrac{9}{10} \cdot 0.45 > 0.1$, with this probability.

To complete the proof, we apply the union bound over the $k!$ permutations,
and conclude that with probability at least $1-k! \cdot 2^{-\Omega(\sqrt{n} \cdot 2^k)} > 1-2^{-\Omega(\sqrt{n})}$
all distances would be at least $0.1$ as required.
\end{proof}

\subsection{Finding the asymmetric sets}\label{sec:psfs-find-asym-sets}

The asymmetric sets in the partition are those with non-zero
symmetric influence. The following algorithm estimates the symmetric
influence of a given set efficiently.

\begin{algorithm}[tbh]
  \caption{\textsc{Estimate-Symmetric-Influence}$(f, J, \delta, \eta)$}
  \begin{algorithmic}[1] \label{alg:est-syminf}
   \STATE Set $q = \lceil \tfrac {\ln 2/\eta}{2\delta^2} \rceil$ and $X = 0$.
   \FOR{$i=1$ to $q$}
    \STATE Pick a random input $x \in \ZZ_2^n$ and a random permutation $\pi \in \calS_J$.
    \STATE Increase $X$ by 1 if $f(x) \neq f(\pi x)$.
   \ENDFOR
   \STATE Return $X/q$.
  \end{algorithmic}
\end{algorithm}

\begin{prop}
For any function $f$, a set of variables $J$ and two constants
$\delta, \eta \in (0,1)$, the algorithm
\textsc{Estimate-Symmetric-Influence}$(f, J, \delta, \eta)$
returns a value within distance $\delta$ of $\SymInf_f(J)$ with
probability at least $1-\eta$, by performing $O(\delta^{-2} \log
1/\eta)$ non-adaptive queries to $f$.
\end{prop}

\begin{proof}
The proof is a direct application of the Chernoff bound as
$\Exp{X/q} = \SymInf_f(J)$ and we deviate by more than
$\delta$ with probability at most $2\exp(-2\delta^2q)$.
\end{proof}

Since our goal is to identify the asymmetric sets in the presence of
noise, we rely on the following proposition. Similar to influence,
if two functions are $\eps$-close, the symmetric influence of every
set deviates by at most $2\eps$. We omit the proof as it is identical
to that of Proposition~\ref{prop:inf-distance}.

\begin{prop}\label{prop:syminf-distance}
Let $f$ and $g$ be any two functions which are $\eps$-close. Then
for every set $J \subseteq [n]$ of variables, $| \SymInf_f(J) -
\SymInf_g(J) | \leq 2\eps$.
\end{prop}

The first part in correcting partially symmetric functions is
identifying the asymmetric sets. This algorithm is identical to the
one for identifying the influencing sets in juntas, with the only
difference of invoking \textsc{Estimate-Symmetric-Influence} instead
of \textsc{Estimate-Influence}. We require our input function to
satisfy the condition of Proposition~\ref{prop:psfs-high-syminf} in
order to correctly identify the asymmetric sets in spite of
the noise.

The algorithm \textsc{Find-Asymmetric-Sets} is given a function $f$,
a partition of the variables $\calI$, a size parameter $k$ and a
noise parameter $\eps$. The algorithm returns all parts in the
partition $\calI$ which it considered as asymmetric with respect to
$\eps$.

\begin{algorithm}[tbh]
  \caption{\textsc{Find-Asymmetric-Sets}$(f, \calI = \{I_1, \ldots, I_s\}, k, \eps)$}
  \begin{algorithmic}[1] \label{alg:find-asym-sets}
   \STATE Fix $r = \lceil 12k \ln s \rceil $ and $S = [s]$.
   \FOR{each of $r$ rounds}
    \STATE Pick a random subset $T \subseteq [s]$ by including each index independently with probability $1/k$.
    \STATE Define $J = \cup_{i \in T} I_i$ to be the union of sets in $\calI$ according to
    the indices of $T$.
    \STATE If \textsc{Estimate-Symmetric-Influence}$(f, J, \epsilon, 1/20r) \leq 3\eps$, set $S = S \setminus T$.
   \ENDFOR
   \STATE Return $\{I_i\}_{i \in S}$
  \end{algorithmic}
\end{algorithm}

\begin{lemma}\label{lem:find-asym-sets}
Let $f$ be an $(n-k)$-symmetric function for $k=o(n/\log n)$, which
satisfies the following conditions. The asymmetric variables of $f$
are separated by the partition $\calI$, which is of size $|\calI| =
O(k^2)$, $|\calI| > 5$. Additionally, the symmetric influence of
every asymmetric variable and a symmetric variable is at least
$6\eps$ for some $\eps>0$. Then for every function $g$ which is
$\eps$-close to $f$, \textsc{Find-Asymmetric-Sets}$(g, \calI, k,
\eps)$ returns exactly the $k$ sets containing the asymmetric
variables of $f$ with probability at least $9/10$, by performing
$O(k \log^2 k/\eps^2)$ non-adaptive queries to $g$.
\end{lemma}

\begin{proof}
The proof of the above lemma is identical to that of
Lemma~\ref{lem:find-inf-sets}, where we replace influence with
symmetric influence when applicable.
Proposition~\ref{prop:syminf-distance} guarantees that in the noisy
function $g$ we would still be able to distinguish the asymmetric
sets. Additionally, by the limitation over $k$, with high
probability in all rounds we would have at least one symmetric
variable, which guarantees the high symmetric influence (assuming
there is also an asymmetric variable).
\end{proof}

\subsection{The algorithm}\label{sec:psfs-alg}

The second step of the algorithm is recovering the correct order of
the $k$ asymmetric sets. As explained before, it is reasonable to
query the input function over inputs which are constant over the
various parts of our random partition. In this case however, we also
care about the Hamming weight of the inputs we query. Due to this,
we cannot allow our queries to be consistent over all the parts, but
rather we will dedicate one part for adjusting the distribution of
the Hamming weight. We name this special part \emph{workspace}, and
we choose it arbitrarily from our random partition.

\begin{definition}
Let $\calI$ be some partition of $[n]$ into an odd number of parts
and let $W \in \calI$ be the workspace. Define the distribution
$\calDIW$ over $\ZZ_2^n$ to be as follows. Pick a random Hamming
weight $w$ according to the binomial distribution $B(n, 1/2)$ and
output, if exists, a random $x \in \ZZ_2^n$ of Hamming weight $|x| =
w$ such that for every $I \in \calI \setminus \{W\}$, either $x_I
\equiv 0$ or $x_I \equiv 1$. When no such $x$ exists, return the all
zeros vector.
\end{definition}

The above distribution, together with the random choice of the
partition and workspace, satisfies the following two important
properties. The first, being close to uniform over the inputs of the
function. The second, having a marginal distribution over the inputs
to the core of a partially symmetric function close to
\emph{typical}. A typical distribution over the core of an
$(n-k)$-symmetric function is the product of a uniform distribution
over $\ZZ_2^k$ and $B(n-k, 1/2)$. These properties are formally
written here as Proposition~\ref{prop:psfs-sample} which originally
appeared in~\cite{BWY}.

\begin{prop}[\cite{BWY}]\label{prop:psfs-sample}
Let  $J=\{j_1, \ldots, j_k\} \subseteq [n]$ be a set of size $k$,
and $s = \Omega(k^2)$ be odd. If $y \sim \calDIW$ for a random
partition $\calI$ of $[n]$ into $s$ parts and a random choice of the
workspace $W \in \calI$, then
\begin{itemize}
\item $y$ is $o(1/n)$-close to being uniform over $\ZZ_2^n$, and
\item $(y_J, |y_{\overline{J}}|)$ is $(k/s + o(1))$-close to being distributed uniformly over
  $\ZZ_2^k$ and binomial over $\{0, 1, \ldots, n-k\}$,
  for a fixed partition $\calI$ which separates the variables of $J$.
\end{itemize}
\end{prop}

\begin{remark}
Our algorithm will choose the partition $\calI$ and workspace $W$
not entirely at random, but rather it will consider the location of
zeros and ones in the input $x$. However, the above proposition
still holds in this scenario, assuming we partition each of the two
sets into $r$ parts at random.
\end{remark}

We can now describe the full algorithm, followed by its analysis.
Recall that our goal is to locally correct a given function $g$,
which is $\eps$-close to $f_{\sigma}$, by returning the value
$f_{\sigma}(x)$ for the given input $x$. Additionally, $f$ is an
$(n-k)$-symmetric function whose core is known to the algorithm, and
we can and will restrict ourselves to such functions which satisfy
typical conditions (namely Propositions~\ref{prop:psfs-high-syminf}
and \ref{prop:psfs-no-iso}). Notice that in our scenario, the
smaller $\eps$ is the easier it is to correct so it suffices to show
that for $\eps=0.001$, the algorithm succeeds with good probability
and with the requested query complexity.

\begin{algorithm}[tbh]
  \caption{\textsc{Locally-Correct-Partially-Symmetric-Function}$(f_{\core}, k, g, x)$}
  \begin{algorithmic}[1]
   \STATE Fix $s = 100k^2$ and $r = 2500 \lceil k\log k \rceil $.
   \STATE Choose a random workspace $W \subseteq [n]$ by including each
   $i \in [n]$ into $W$ with probability $1/(2s+1)$.
   \FOR{$j \in \{0,1\}$}
     \STATE Let $X_j = \{ i \in [n] \setminus W \mid x_i = j \}$.
     \STATE Randomly partition $X_j$ into $\calI_j$, consisting of (potentially) $s$ parts.
   \ENDFOR
   \STATE Define $\calI = \calI_0 \cup \calI_1 \cup \{W\}$ to be our partition.
   \STATE Invoke \textsc{Find-Asymmetric-Sets}$(g, \calI, k, 0.01)$ and assign the result into $\calJ = \{ I_{a_1}, I_{a_2}, \ldots \}.$
   \STATE If $ | \calJ | \neq k$, or if $\calJ$ contains $W$ or an empty set, return 0.
   \STATE Let $B = (b_1, \ldots, b_k)$ be an arbitrary ordered set such that $b_i \in I_{a_i}$ for every $i \in [k]$.
   \STATE For every $\ell \in [r]$, randomly sample $y^\ell \sim \calDIW$ and query $g(y^\ell)$.
   \STATE Let $\pi \in \calS_k$ be the permutation which maximizes the number of indices $\ell$
   for which \\ $g(y^\ell) = f_{\core}(y^\ell_{B_{\pi}},
   |y^\ell_{\overline{B}}|)$, where $B_{\pi} = (b_{\pi(1)}, \ldots,
   b_{\pi(k)})$.
   \STATE Return $f_{\core}(x_{B_{\pi}}, |x_{\overline{B}}|)$
  \end{algorithmic}
\end{algorithm}

\begin{proof}[Proof of Theorem~\ref{thm:psfs}]
The queries of the algorithm are performed by invoking
\textsc{Find-Asymmetric-Sets} and by sampling according to
$\calDIW$. In both cases the queries are non-adaptive and depend
only on our choice of the partition $\calI$ and workspace $W$. The
total number of queries performed is $O(k\log^2k)$, as required.

Let $f$ be an $(n-k)$-symmetric function which satisfies the conditions of
Propositions~\ref{prop:psfs-high-syminf} and \ref{prop:psfs-no-iso}.
As each condition is satisfied with probability at least
$1-2^{-\Omega(\sqrt{n})}$, it suffices to show the algorithm succeeds with
high probability when applied to these functions.
Moreover, when $k = \Omega(n/\log n)$, Theorem~\ref{thm:psfs} trivially holds
by applying the standard non-adaptive algorithm
of querying $O(n \log n)$ uniform queries (given that Proposition~\ref{prop:psfs-no-iso} is satisfied).
Therefore, we assume $k = o(n / \log n)$ throughout the rest of the proof.

The success of the algorithm depends on the following three events.
First, we need the partition $\calI$ to separate all the
asymmetric variables, and that none of them would belong to the
workspace. An asymmetric variable would be chosen to the workspace
with probability at most $k/s$. At each set of variables, $X_0$ and
$X_1$, there are at most $k$ asymmetric variables and therefore each
would have a collision with probability at most ${k \choose 2}/s$.
Therefore, we would have a bad partition and workspace with
probability at most $(k+2{k \choose 2})/s < 1/30$ for our choice of
$s$.

The second event is identifying the correct sets, those for which
$f_{\sigma}$ has asymmetric variables. Assuming the first event
occurred, by Lemma~\ref{lem:find-asym-sets} we will identify
precisely the asymmetric sets with probability at least $9/10$.

The third and last event which remains is correctly choosing the
permutation $\pi$. Here we rely on the fact that the core of $f$ is
not close to any non trivial permutation of itself. Since our
estimates are performed using queries according to $\calDIW$, we
need to estimate how accurate they are. By
proposition~\ref{prop:psfs-sample}, as each sample $y$ is
distributed $o(1/n)$-close to uniform, the distance between $g(y)$
and $f_{\sigma}(y)$ is at most $0.001 + o(1/n)$ over the choice of
$\calI$ and $W$. By applying Markov's inequality, our partition and
workspace will satisfy $\Pr_{y\sim \calDIW}[g(y) \neq f_{\sigma}(y)]
\leq 0.02$ with probability at least $9/10$.

Assume from this point on that the partition and workspace satisfy
the above inequality. By the second part of
proposition~\ref{prop:psfs-sample}, the marginal distribution of $y$
over the asymmetric variables (and in particular over the asymmetric
sets) would be $0.01$-close to uniform in total variation distance.
Therefore, when $\pi \in \calS_k$ is not the correct permutation we
have
$$
\Expx{y \sim \calDIW} { g(y) = f_{\core}(y_{B_{\pi}},
|y_{\overline{B}}|) }
  \leq 1 - 0.1 + 0.02 + 0.01 = 0.93 \ ,
$$
where $B$ and $B_{\pi}$ are defined as in the algorithm. For the
correct permutation, however, this expectation is at least
$0.97$. Since our queries are independent from one another, we can
apply the Chernoff bound. When performing $r$ such queries, we 
deviate from the expectation by at least 0.02 with probability at
most $\exp(-2\cdot 0.02^2 \cdot r) = \exp(-2\lceil k \log k
\rceil)$. By the union bound, our estimation for all $k!$
permutations is within 0.02 from the expectation with
probability at least $9/10$. Combining this with the probability
that the partition and workspace are good, the correct
permutation $\pi$ is chosen with probability at least $4/5$.

The failure probability of our algorithm can be bounded by the
probability that one of the above events will not occur. Therefore,
the algorithm will fail with probability at most $1/30 + 1/10 + 4/5
= 1/3$, meaning it will return the correct answer with probability
at least $2/3$, as required.
\end{proof}

\begin{remark}
As in Theorem~\ref{thm:juntas}, this proof also holds for partially
symmetric functions which do have some non-trivial isomorphisms, but
for simplicity we do not consider this here.
\end{remark}

\section{Conclusions and open problems}
\label{sec:conclusions} 

In the previous sections we have shown that most juntas and partially symmetric
functions can be efficiently locally corrected from a constant error rate.
Although Proposition~\ref{prop:impossible-juntas-psfs} indicates that not every junta
or partially symmetric function satisfy this, it provides no insight on whether other
functions can be locally corrected efficiently under similar conditions.

To better understand the question of which functions can be locally corrected efficiently from
a constant noise rate, we present the proof of Theorem~\ref{thm:most-functions}.
This theorem does not provide a characterization of these locally correctable functions,
but rather indicates that most functions do not fall into this category. In order to prove the theorem, we
use the main results of~\cite{AB}, where it is shown that testing isomorphism to
almost every function requires at least a linear number of non-adaptive queries.

\begin{prop}[\cite{AB}]
For almost every Boolean function $f$ the following holds.
Let $Q = Q_b \cup Q_u$ be a set of $\tfrac{n}{100}$ balanced queries
(of Hamming weight between $\tfrac{n}{3}$ and $\tfrac{2n}{3}$) and $\tfrac{n}{100}$ unbalanced queries.
Over the choice of a random isomorphism $f_{\sigma}$ of $f$, with probability $1-o(1)$
over the output of $f_{\sigma}$ at $Q_u$, every possible outcome $r \in \ZZ_2^{|Q_b|}$
of querying $f_{\sigma}$ at $Q_b$ is obtained with probability $(1 \pm \tfrac{1}{3})2^{-|Q_b|}$.
\end{prop}

\begin{proof}[Proof of Theorem~\ref{thm:most-functions}]
The above proposition is essentially everything needed for completing the proof.
Given a random Boolean function which satisfies the condition of the above proposition,
we can request one to locally correct some arbitrary balanced input.
Let us assume there exists an algorithm that can perform this task with good probability,
using less than $\tfrac{n}{100}$ non-adaptive queries. We can extend its query set to consist of
$\tfrac{n}{100}$ balanced and $\tfrac{n}{100}$ unbalanced queries, including the input we asked to correct
(obviously this cannot reduce its success probability).

By the above proposition, with high probability over the
output of the function over the unbalanced queries, the distribution of the function over the balanced queries
is very close to uniform. In fact, for every possible output over all the balanced queries except the
questioned input, the marginal distribution that remains is $1/3$-close to uniform,
and hence one would not be able
to predict its value with probability $\geq 2/3$. Notice that for the purpose of the analysis, we assume
no noise was used, where in practice we simply override the value of the questioned input with
the value 0.
\end{proof}

Like juntas and partially symmetric functions, there are several other families it would be interesting
to investigate with respect to local correction. One natural family is that of low degree polynomials.
It was already shown (see \cite{AW}, \cite{AKKLR}) that these can be
locally corrected using an exponential number of queries,
when the noise rate is also exponentially small (in term of the degree of the polynomial).
However, whether this can be done when the noise rate is constant is yet unknown.

\begin{question*}
Fix some constant $\eps > 0$ and a degree $d > \lceil \log_2 1/\eps \rceil $.
Do most degree $d$ polynomials over $n \geq d$ variables require $\Omega(\log n)$ queries to be
locally corrected from $\eps$ noise rate?
\end{question*}

Low degree polynomials have a rather rigid structure, although they may still be far from symmetric
with respect to any subset of the input variables (even for very small degrees).
When the noise rate is larger than $2^{-d}$ (where $d$ is the degree of the polynomial), the noisy
function can actually be another polynomial of the same degree. In such a case, making use of the structure of the polynomial
seems trickier, as one cannot simply correct according to it. Given the natural asymmetric nature of most
low degree polynomials, it may be the case that they are similar to random functions, in the sense
that one needs a number of queries which grows with $n$ to perform this task, even for constant degrees.

\end{document}